\documentclass[11pt]{article}

\usepackage[utf8]{inputenc}
\usepackage{graphicx}
\usepackage{xcolor}
\usepackage{amsmath}
\usepackage{amssymb}
\usepackage{amsthm}
\usepackage{mathtools}
\usepackage{bm}
\usepackage{cuted}
\usepackage{lineno}
\usepackage{caption}
\usepackage{subcaption}
\usepackage{algorithm}
\usepackage{algpseudocode}
\usepackage{amsmath}
\usepackage{blkarray}

\algtext*{EndFor}
\algtext*{EndIf}
\algtext*{EndFunction}

\newtheorem{theorem}{Theorem}
\newtheorem{lemma}{Lemma}
\newtheorem{corollary}{Corollary}
\newtheorem{proposition}{Proposition}

\usepackage[letterpaper]{geometry}
\title{A Geometric View of Combinatorial Fiedler Theory}

\author{
José Fern\'andez Goycoolea\thanks{\emph{Corresponding author}. Departamento de Matem\'atica y F\'isica, Universidad de Magallanes, Avenida Bulnes 01855, Punta Arenas, Chile,
{\tt jose.fernandezg@umag.cl}, https://orcid.org/0000-0001-5349-4348.} 
\and
Andrea de las Heras-Parrilla\thanks{Universidad Francisco de Vitoria, Spain, and Departament de Matem\`atiques, Universitat Polit\`ecnica de Catalunya, Spain. 
Supported by project PID2023-150725NB-I00 funded by MICIU/AEI/10.13039/501100011033. {\tt andrea.delasheras@ufv.es}
https://orcid.org/0000-0002-7219-9771.
}
\and
Luis H. Herrera\thanks{Departamento de Inform\'atica y Computaci\'on, Universidad Tecnol\'ogica Metropolitana, 
Jos\'e Pedro Alessandri 1242, \~{N}u\~{n}oa, Santiago de Chile 7800002, Región Metropolitana, 
Chile, 
{\tt luis.herrerab@utem.cl}, https://orcid.org/0000-0001-7338-7611.}
\and
Clemens Huemer\thanks{Departament de Matem\`atiques, Universitat Polit\`ecnica de Catalunya, Spain. Supported by project PID2023-150725NB-I00 funded by MICIU/AEI/10.13039/501100011033. {\tt clemens.huemer@upc.edu}, https://orcid.org/0000-0001-7557-0823.
}
\and
Carlos Seara\thanks{Departament de Matem\`atiques, Universitat Polit\`ecnica de Catalunya, Spain. Supported by project PID2023-150725NB-I00 funded by MICIU/AEI/10.13039/501100011033. {\tt carlos.seara@upc.edu}, https://orcid.org/0000-0002-0095-1725.
}}

\begin{document}
%\linenumbers
\maketitle

\begin{abstract}
    Recently, Andrade and Dahl introduced combinatorial Fiedler theory by studying a parameter $b(G)$ defined as the $\ell_1$-analog of the Rayleigh quotient minimization characterization of the algebraic connectivity of a graph $G=(V,E)$. In this work, we study the corresponding maximization problem, which plays the role of the $\ell_1$-analog of the largest Laplacian eigenvalue. We show that the new parameter $B(G)$ associated with this maximization problem admits a simple exact description: it is the average of the two largest vertex degrees of $G$.

    A unified combinatorial treatment of the minimization and maximization problems is presented first. Later, both optimization problems are reinterpreted in a geometrical setting. The feasible set is identified with a $(n-2)$-dimensional cuboctahedron shell where $n=|V|$. Additional structure is presented for this polyhedron, including the fact that maximizing solutions arise at its vertices and minimizing solutions arise at the centers of its facets.

    Finally, we analyze the number of optimal vectors for $b(G)$ and $B(G)$ for several graph families. Although the value of $B(G)$ is determined by the two largest degrees, we prove that counting the vectors that attain this value is actually $\#\mathrm{P}$-complete.
\end{abstract}

\bigskip


\newpage

\section{Introduction}\label{sec:intro}

    Spectral graph theory studies graphs through the eigenvalues and eigenvectors of matrices naturally associated with them, such as the adjacency and Laplacian matrices. An early central example of this approach is Fiedler's algebraic connectivity $a(G)$~\cite{fiedler1973algebraic}, the second-smallest Laplacian eigenvalue, which admits a Rayleigh quotient characterization as a minimization problem; more precisely, let $G=(V,E)$ be a simple graph, then
    \[a(G) = \min_{\bm x\in\mathbf{R}^V}
        \left\{\sum_{uv\in E}(x_u-x_v)^2~:~\sum_{v\in V}x_v=0~\wedge~\sum_{v\in V}x_v^2=1\right\};\]
    see~\cite{brouwer2011spectra} for a systematic treatment. The corresponding eigenvectors, now called Fiedler vectors, assign a real coordinate to each vertex and can therefore be interpreted as one-dimensional representations of the graph. This interpretation underlies their use in graph partitioning and graph drawing~\cite{hall1970r,koren2005drawing}.

    A recent contribution of Andrade and Dahl~\cite{andrade2024combinatorial} initiates what they call \emph{combinatorial Fiedler theory}, replacing the $\ell_2$-norm in the classical characterization of the $a(G)$ parameter by the $\ell_1$-norm. They introduce the parameter
    \begin{equation}\label{eq:b(G)}
        b(G) = \min_{\bm x\in\mathbf{R}^V}
        \left\{\sum_{uv\in E}|x_u-x_v|~:~\sum_{v\in V}x_v=0~\wedge~\sum_{v\in V}|x_v|=1\right\},
    \end{equation}
    which may be regarded as an $\ell_1$ analog of the algebraic connectivity.\footnote{Note that throughout this work the indexation of the vectors $\bm x\in\mathbf{R}^V$ is done directly over the elements of $V$ to unclutter the notation as in the presentation by Brouwer and Haemers~\cite{brouwer2011spectra}.} The resulting problem preserves the flavor of classical theory, but shifts it from a primarily linear-algebraic setting to a discrete combinatorial one, closely related to sparsest cuts. 

    The main purpose of this paper is to extend this theory by studying its natural maximization counterpart. We introduce the parameter
    \begin{equation}\label{eq:B(G)}
        B(G)=\max_{\bm x\in\mathbf{R}^V}
        \left\{\sum_{uv\in E}|x_u-x_v|~:~\sum_{v\in V}x_v=0~\wedge~\sum_{v\in V}|x_v|=1\right\}.
    \end{equation}
    Both parameters are defined as optimization problems over the same feasible region, with the same objective function. From the perspective of the classical \(\ell_2\)-theory, they may be regarded as the \(\ell_1\)-analogs of the two extremal Laplacian eigenvalues \(\lambda_2\) and \(\lambda_n\). While the minimization problem already admits a rich combinatorial interpretation~\cite{andrade2024combinatorial,KannanRoy2026}, the maximization problem appears not to have been systematically investigated before. One of our main results shows that \(B(G)\) has a simple exact form: if \(d_1\) and \(d_2\) are the two largest vertex degrees of \(G\), with possibly \(d_1=d_2\), then \(B(G)=\frac12(d_1+d_2)\).
    Degree-sum expressions are familiar in the study of extremal Laplacian parameters, for instance in upper bounds for the largest Laplacian eigenvalue~\cite{anderson1985eigenvalues,zhang2011laplacian}. Here, however, the degree average equals \(B(G)\).
    
    As a first step, we extend some of the results of Andrade and Dahl so that they apply simultaneously to the minimization and maximization problems. In particular, both extrema admit solutions whose positive entries are all equal and whose negative entries are all equal. Thus, although the feasible region has infinite cardinality, both optimization problems can be reduced to a finite set \(\mathcal{G}\subset\mathbf{R}^V\) of structured vectors. This set has a natural combinatorial interpretation: its elements are in bijection with the ordered quasi-bipartitions \((P,N)\) of the vertex set, where \(P\) records the positive coordinates, \(N\) records the negative coordinates and the remaining vertices have coordinate zero. Under this correspondence, the objective function is expressed in terms of relative cut sizes, giving cut-based formulas for both \(b(G)\) and \(B(G)\). For the minimization problem, this recovers the characterization of Andrade and Dahl~\cite{andrade2024combinatorial}; for the maximization problem, it leads to the degree formula stated above.

    The second main contribution of our work is to reinterpret these optimization problems geometrically. The common feasible region can be viewed as the relative boundary of a hyperplane section of the \(\ell_1\) unit ball. This places the problem in a natural polyhedral setting related to the \(d\)-dimensional cuboctahedron~\cite{doehlert1972experimental} and to root polytopes~\cite{ardila2010root}. 

    Related \(\ell_1\)-spectral approaches to graph partitioning have also been studied in the context of Cheeger cuts~\cite{chung1997spectral}. In particular, Chang, Shao and Zhang use a cell decomposition of the feasible set for the graph \(1\)-Laplacian Cheeger problem, noting that the objective function is convex on each cell~\cite{chang2016one}. This is close in spirit to the approach taken here: we also decompose an \(\ell_1\)-type feasible region into cells adapted to the objective function; in our setting, this refinement makes the objective affine on each cell.

    The paper is organized as follows. Section~\ref{sec:preliminaries} develops the finite combinatorial formulation of the two optimization problems and proves the degree formula for \(B(G)\). Section~\ref{sec:geometry} studies the geometry of the common feasible region, its cell decomposition, and the induced structure on \(\mathcal{G}\) to end with alternative geometric proofs of some of the main results. Section~\ref{sec:examples} applies the studied results to several graph families and counts the corresponding \(\ell_1\)-Fiedler vectors of $\mathcal{G}$ in these examples. Section~\ref{sec:hardness} proves the \(\#\mathrm{P}\)-completeness of counting the vectors in \(\mathcal{G}\) that attain \(B(G)\).

\section{Preliminaries}\label{sec:preliminaries}   
    Notice that both optimization problems, of Equation~\eqref{eq:b(G)} and Equation~\eqref{eq:B(G)}, are defined over the same \textbf{feasible set} $\mathcal{F}\subset\mathbf{R}^V$ defined by $\mathcal{F}=\{\bm x\in\mathbf{R}^V~:~\sum_{v\in V}x_v=0~\wedge~\sum_{v\in V}|x_v|=1\}$ and over the same \textbf{objective function} $f:\mathbf{R}^V\rightarrow\mathbf{R}_{\geq0}$ given by $f(\bm x)=\sum_{uv\in E}|x_u-x_v|$, so we can rewrite Equations~\eqref{eq:b(G)} and \eqref{eq:B(G)} as
\begin{equation}\label{eq:opt_problems_F}
    b(G)=\min_{\bm x\in\mathcal{F}}f(\bm x),~~B(G)=\max_{\bm x\in\mathcal{F}}f(\bm x).    
\end{equation}
    
    The feasible set can be alternatively understood as the set of vectors whose positive components add to $1/2$ and whose negative components add to $-1/2$. 
    
\begin{lemma}[Andrade and Dahl~\cite{andrade2024combinatorial}]\label{lemm:+1/2_-1/2}
    $\displaystyle\mathcal{F}=\Bigg\{\bm x\in\mathbf{R}^V~:~\sum_{\substack{v\in V\\x_v\geq0}}x_v=\frac1{2}~\wedge~\sum_{\substack{v\in V\\x_v\leq0}}x_v=-\frac1{2}\Bigg\}.$
\end{lemma}
 
    A vector that realizes one of the optimization problems of Equation~\eqref{eq:opt_problems_F} is called an $\ell_1$-Fiedler vector of that problem. That is, for example, that $\bm x\in\mathcal{F}$ is an $\ell_1$-Fiedler vector for $B(G)$ if and only if $f(\bm x)=B(G)$. A graph \(G\) might have many different \(\ell_1\)-Fiedler vectors; understanding how many there are is a natural question, analogous to the role played in the classical \(\ell_2\)-theory by the multiplicity of the Fiedler eigenvalue. We return to these counting questions later in the paper.
    
    Among the $\ell_1$-Fiedler vectors of a graph $G$, those of the form described in the next lemma yield the combinatorial characterization of $b(G)$ and $B(G)$. The following Lemma~\ref{lemm:same_pos_same_neg} corresponds to the first part of Theorem 3.2 in~\cite{andrade2024combinatorial}. The proof given here follows the same general approach as the one in~\cite{andrade2024combinatorial}. We nevertheless include it because it amends the final part of the argument and also extends it to the maximization case.

\begin{lemma}\label{lemm:same_pos_same_neg}
    For both optimization problems of Equation \eqref{eq:opt_problems_F}, there exist $\ell_1$-Fiedler vectors for which each positive entry has the same value, and each negative entry has the same value.
\end{lemma}
    
\begin{proof}
    The existence of an $\ell_1$-Fiedler vector for which all positive entries are equal is proven by contradiction. To do so, we consider an $\ell_1$-Fiedler vector $\bm x\in\mathbf{R}^V$ such that the number $\kappa(\bm x)$ of different positive entries is as small as possible. Then the contradiction argument is constructed in two steps: first, we show that a modification of the vector $\bm x$ depending on a \emph{sufficiently small} $\varepsilon$ can be made so that it remains an $\ell_1$-Fiedler vector. Then we verify that an $\varepsilon$ in the valid range can be selected such that the number $\kappa(\bm x)$ is reduced.
    
    Consider the smallest and largest positive values among the entries of $\bm x$; that is, $m=\min\{x_v : x_v>0\}$ and $M=\max\{x_v : x_v>0\}$. If $m=M$, there is nothing to prove; all positive entries are equal. So assume $m<M$ and define the following partition of $V$ given by the vector $\bm x$:
    \begin{equation}\label{eq:V_lemma_partition}
    \begin{split}
        S_0 & = \{v\in V ~:~ x_v\leq0\}, \\
        S_1 & = \{v\in V ~:~ x_v=m\}, \\
        S_2 & = \{v\in V ~:~ m<x_v<M\}, \\
        S_3 & = \{v\in V ~:~ x_v=M\}. 
    \end{split}
    \end{equation}
    Given an $\varepsilon$ value small in magnitude but possibly positive or negative, we construct the vector $\bm x^\varepsilon\in\mathbf{R}^V$ in which the values of $v\in S_1$ are shifted by $\varepsilon$ and, in turn, the values of each $v\in S_3$ are shifted by an amount that compensates the positive sum in the characterization of $\mathcal F$ of Lemma~\ref{lemm:+1/2_-1/2}. That is 
    \begin{equation}
        x_v^\varepsilon=\left\{\begin{matrix*}[l]
            x_v + \varepsilon & \text{if } v\in S_1, \\
            x_v - \varepsilon|S_1|/|S_3| & \text{if } v\in S_3, \\
            x_v & \text{if } v\in S_0\cup S_2.
    \end{matrix*}\right.    
    \end{equation}
    With these shifts, if $\varepsilon$ is \emph{sufficiently small}, we have
    \begin{equation*}
        \begin{split}
        \sum_{\substack{v\in V\\x_v^\varepsilon\geq0}}x_v^\varepsilon=\sum_{\substack{v\in V\\x_v>0}}x_v^\varepsilon &=    
        \sum_{v\in S_1}x_v^\varepsilon+\sum_{v\in S_2}x_v^\varepsilon+\sum_{v\in S_3}x_v^\varepsilon \\
        &= \sum_{v\in S_1}\big(x_v+\varepsilon\big)+\sum_{v\in S_2}x_v+\sum_{v\in S_3}\left(x_v-\varepsilon\frac{|S_1|}{|S_3|}\right) \\
        &= \sum_{v\in S_1}x_v+\varepsilon|S_1|+\sum_{v\in S_2}x_v+\sum_{v\in S_3}x_v-\varepsilon\frac{|S_1|}{|S_3|}|S_3|=\sum_{\substack{v\in V\\x_v>0}}x_v=\frac{1}{2},
        \end{split}
    \end{equation*}
    and therefore, $\bm x^\varepsilon\in\mathcal F$. In this calculation, \emph{sufficiently small} explicitly means staying in the range
    \[-m<\varepsilon<\min\Big\{\min\{x_v:v\in S_2\cup S_3\}-m,~M-\max\{x_v:v\in S_1\cup S_2\}
    \Big\},\]
    as for this range, the partition for $\bm x$ of Equation~\eqref{eq:V_lemma_partition} stays invariant for $\bm x^\varepsilon$. Now, for $\varepsilon$ in this range, the difference in cost function can be described as a function of $\varepsilon$:
    \[\Delta(\varepsilon)=f(\bm x)-f(\bm x^\varepsilon)=\sum_{uv\in E}|x_u-x_v|-\sum_{uv\in E}|x_u^\varepsilon-x_v^\varepsilon|=\sum_{uv\in E}\Big(\underbrace{|x_u-x_v|-|x_u^\varepsilon-x_v^\varepsilon|}_{\Delta_{uv}(\varepsilon)}\Big),\]
    where $\Delta_{uv}(\varepsilon)=0$ for all $uv\in E$ except in the following four cases
    \begin{equation}\label{eq:delta_uv_cases}
        \begin{split}
            \text{(a) If } & u\in S_1\wedge v\in S_0 \text{ then } \Delta_{uv}(\varepsilon) = -\varepsilon, \\
            \text{(b) If } & u\in S_1\wedge v\in S_2 \text{ then } \Delta_{uv}(\varepsilon) = \varepsilon, \\
            \text{(c) If } & u\in S_1\wedge v\in S_3 \text{ then } \Delta_{uv}(\varepsilon) = \varepsilon\big(1+|S_1|/|S_3|\big), \\
            \text{(d) If } & u\in S_3\wedge v\in S_0\cup S_2 \text{ then } \Delta_{uv}(\varepsilon) = \varepsilon|S_1|/|S_3|. \\
        \end{split}
    \end{equation}
    Labeling by $N_a$, $N_b$, $N_c$, and $N_d$ the number of edges $uv\in E$ in each of the corresponding cases we can write
    \[\Delta(\varepsilon)=\left(-N_a+N_b+N_c+(N_c+N_d)\frac{|S_1|}{|S_3|}\right)\varepsilon=\eta\varepsilon.\]
    
    The key insight of the first part of the argument is that the number $\eta\in\mathbf{R}$ is fixed for $\varepsilon$ in the sufficiently small range; however, $\varepsilon$ can take both positive and negative values, so the only possibility consistent with the optimality of $\bm x$ is that $\eta=0$. Namely, if the optimization problem of $b(G)$ is considered, an $\varepsilon$ of the same sign as $\eta$ could be chosen, giving $\Delta(\varepsilon)>0$ and then $f(\bm x^\varepsilon)<f(\bm x)$; for $B(G)$, an $\varepsilon$ of opposite sign to $\eta$ would make $\Delta(\varepsilon)<0$ and then $f(\bm x^\varepsilon)>f(\bm x)$. Therefore, in both cases $\eta=0$ necessarily and $f(\bm x)=f(\bm x^\varepsilon)$ so $\bm x^\varepsilon$ is also $\ell_1$-Fiedler.

    Now, the idea for the second part of the argument is to increase $\varepsilon>0$ up to the limit of the sufficiently small range. That could be where the smallest positive coordinate of $\bm x$ has been increased up to the second smallest $\varepsilon=\min\{x_v:v\in S_2\cup S_3\}-m$, or where the largest coordinate has been reduced down to the second largest $\varepsilon=M-\max\{x_v:v\in S_1\cup S_2\}$ (or both things occur simultaneously). The key insight here is to realize that at this point, still $\bm x^\varepsilon\in\mathcal{F}$ and the $\Delta_{uv}(\varepsilon)$ values of Equation~\eqref{eq:delta_uv_cases} are the same, so $\bm x^\varepsilon$ is an $\ell_1$-Fiedler vector. However, $\kappa(\bm x^\varepsilon)<\kappa(\bm x)$, a contradiction.

    Finally, starting from this $\ell_1$-Fiedler vector, where all positive entries are the same, one can follow the analogous argument for the negative entries of the vector to produce the desired $\ell_1$-Fiedler vector. 
\end{proof}
    
    Consider the subset $\mathcal{G}\subset\mathcal{F}$ consisting of the vectors for which all positive entries have the same value and all negative entries have the same value; that is, $\mathcal{G}=\{\bm x\in\mathcal{F}\subset\mathbf{R}^V~:~x_u x_v>0 \Rightarrow x_u=x_v\}$. While $\mathcal{F}$ has infinite cardinality, $\mathcal{G}$ is finite, with $|\mathcal{G}|=3^n-2^{n+1}+1$. Nevertheless, Lemma~\ref{lemm:same_pos_same_neg} shows that the optimization problems in Equation~\eqref{eq:opt_problems_F} can be restricted to this set.
\begin{equation}\label{eq:opt_problems_G}
    b(G)=\min_{\bm x\in\mathcal{G}}f(\bm x),~~B(G)=\max_{\bm x\in\mathcal{G}}f(\bm x).
\end{equation}

    The combinatorial description of the parameters $b(G)$ and $B(G)$ comes from the fact that, for vectors in $\mathcal{G}$, the objective function $f(\bm x)=\sum_{uv\in E}|x_u-x_v|$ can be rewritten as the average of the relative cut sizes of a quasi-bipartition of $V$. This relies on the natural bijection between $\mathcal{G}$ and the set of ordered quasi-bipartitions of $V$, which is established in the next proposition. We define the relevant concepts first.
    
    For a subset $S\subset V$ of the vertex set of any graph $G=(V,E)$, the \textbf{cut} induced by $S$ is the set of edges $\delta(S) = \{uv\in E ~:~ u\in S \wedge v\in S^c\}$ that join a vertex of $S$ with a vertex of its complement $S^c=V\setminus S$. For a non-empty proper subset $S\subset V$, the \textbf{relative cut size} of $S$ is defined by the quotient
\begin{equation}\label{eq:xi(S)}
    \xi(S) = |\delta(S)|~/~|S|.
\end{equation}

    A quasi-bipartition of a set is simply a bipartition with the covering requirement relaxed. Thus, a pair $(S_1,S_2)$ of subsets of $V$ is an \textbf{ordered quasi-bipartition} of $V$ if $S_1,S_2\neq\varnothing$ and $S_1\cap S_2=\varnothing$. Let $\Gamma$ be the set of all ordered quasi-bipartitions of $V$, then we have the following result.

\begin{proposition}\label{prop:sum_xis}
    There is a natural bijection between the sets $\Gamma$ and $\mathcal{G}$ that sends each pair $(P,N)\in\Gamma$ to the point $\bm g_{\scriptscriptstyle P,N}=(g_v)_{v\in V}\in\mathcal{G}$ with components given by
    \[g_v=\left\{\def\arraystretch{1.4}\begin{array}{rl}
         \frac1{2|P|}, & \text{if } v\in P, \\
        -\frac1{2|N|}, & \text{if } v\in N, \\
            0, & \text{if } v\in V\setminus(P\cup N).
    \end{array}\right.\]
    Moreover,
    \[f(\bm g_{\scriptscriptstyle P,N})=\frac12\big(\xi(P)+\xi(N)\big).\]
\end{proposition}

\begin{proof}
    We prove that the map $\Gamma\rightarrow\mathcal{G}$ defined by the proposition is in fact a bijection by exposing its inverse explicitly. Before that, note that for any $(P,N)\in\Gamma$ the image $\bm g_{\scriptscriptstyle P,N}=(g_v)_{v\in V}$ is in fact in $\mathcal{G}$ because all its positive (and negative) components are equal by definition; and $\bm g_{\scriptscriptstyle P,N}\in\mathcal{F}$ by Lemma~\ref{lemm:+1/2_-1/2} since
    \[\sum_{\substack{v\in V\\g_v\geq0}}g_v=\sum_{v\in P}g_v=\frac{1}{2|P|}|P|=\frac12
    ~~~~~~\text{and}~~~~~
    \sum_{\substack{v\in V\\g_v\leq0}}g_v=\sum_{v\in N}g_v=-\frac{1}{2|N|}|N|=-\frac12.\]
    Now, let $\bm x\in\mathcal{G}$ and define $P=\{v\in V~:~x_v>0\}$ and $N=\{v\in V~:~x_v<0\}$. All coordinates $x_v$ corresponding to $v\in P$ are equal and sum to $\frac12$ because $\bm x$ is in $\mathcal{G}$, so necessarily $x_v=1/(2|P|)$ if $v\in P$. Similarly $x_v=-1/(2|N|)$ if $v\in N$, so $(P,N)$ is the desired pre-image of $\bm x$ under the map. That is $\bm x=\bm g_{\scriptscriptstyle P,N}$, therefore the map is a bijection.

    To see the last part let $(P,N)\in\Gamma$ and set $Z=V\setminus(P\cup N)$ so that $P^c=N\cup Z$ and $N^c=P\cup Z$. Note that only edges with endpoints in different parts contribute to the sum of $f$. Thus 
    \begin{align*}
            f(\bm{g}_{\scriptscriptstyle P,N})=\sum_{uv\in E}|g_u - g_v|
            &=  \sum_{\substack{uv\in E \\ u\in P,v\in N}}|g_u - g_v| + \sum_{\substack{uv\in E \\ u\in P,v\in Z}}|g_u - g_v| + \sum_{\substack{uv\in E \\ u\in Z,v\in N}}|g_u - g_v| \\
             &=  \sum_{\substack{uv\in E \\ u\in P,v\in N}}\left(\frac1{2|P|}+\frac1{2|N|}\right) + \sum_{\substack{uv\in E \\ u\in P,v\in Z}}\frac1{2|P|} + \sum_{\substack{uv\in E \\ u\in Z,v\in N}}\frac1{2|N|} \\
             &= \frac1{2|P|}\Bigg(\sum_{\substack{uv\in E \\ u\in P,v\in N}}1+\sum_{\substack{uv\in E \\ u\in P,v\in Z}}1\Bigg) + \frac1{2|N|}\Bigg(\sum_{\substack{uv\in E \\ u\in P,v\in N}}1+\sum_{\substack{uv\in E \\ u\in Z,v\in N}}1\Bigg) \\
             &= \frac1{2|P|}\sum_{\substack{uv\in E \\ u\in P,v\in P^c}}1 +\frac1{2|N|}\sum_{\substack{uv\in E \\ u\in N,v\in N^c}}1 \\
             &= \frac{|\delta(P)|}{2|P|} + \frac{|\delta(N)|}{2|N|} \\
             &= \frac12\big(\xi(P)+\xi(N)\big).
    \end{align*}
    This completes the proof.
\end{proof}

    From the previous proposition, it follows that the optimization problems of $f(\bm x)=\sum_{uv\in E}|x_u-x_v|$ over vectors of $\bm x\in\mathcal G$ can be reinterpreted as optimization problems of $\frac{1}{2}\big(\xi(P)+\xi(N)\big)$ over ordered quasi-bipartitions $(P,N)\in\Gamma$. This is stated in the next result, of which the statement for $b(G)$ is Theorem 3.2 of Andrade and Dahl~\cite{andrade2024combinatorial}.

\begin{theorem}\label{theo:min_over_oqbp} For any graph $G=(V,E)$:
    \[b(G)=\frac{1}{2} \min_{(P,N)\in\Gamma}\big(\xi(P)+\xi(N)\big)~\text{ and }~B(G)=\frac{1}{2}\max_{(P,N)\in\Gamma}\big(\xi(P)+\xi(N)\big).\]
\end{theorem}

    Up to now, there has been a symmetry between the maximum and minimum optimization problems of $f(\bm{x})$; this symmetry ends here. For a connected graph $G$, no $\ell_1$-Fiedler vector of $b(G)$ contains zero values~\cite{andrade2024combinatorial}; however, for any graph with more than two vertices, there are always  $\ell_1$-Fiedler vectors of $B(G)$ that contain zero values. The following result fully characterizes the value of $B(G)$ in terms of the two highest degrees among the vertices of $G$.

\begin{theorem}\label{theo:two_highest_geg}

    Given a graph $G=(V,E)$, let $v_1,v_2\in V$ be two vertices of the highest degree. That is to say, $\text{deg}_G(v_1)\geq \text{deg}_G(v_2)\geq\text{deg}_G(v)$, for all $v\in V$, $v\neq v_1,v_2$. Then
    \begin{equation*}
        B(G) = \frac1{2}\big(\text{deg}_G(v_1)+\text{deg}_G(v_2)\big).
    \end{equation*}
\end{theorem}

\begin{proof}
    For any subset $S\subset V$ we define $\text{deg}_{max}(S)=\max\{\text{deg}_G(v):v\in S\}$. Then $|\delta(S)|\leq|S|\text{deg}_{max}(S)$. 
    
    Now, let $(S_1,S_2)\in\Gamma$ be any ordered quasi-bipartition of $V$. Then
    \begin{align*}
        \frac12\big(\xi(S_1)+\xi(S_2)\big) = \frac{|\delta (S_1)|}{2|S_1|} + \frac{|\delta (S_2)|}{2|S_2|} 
        &= \frac{|S_2||\delta (S_1)|+|S_1||\delta (S_2)|}{2|S_1||S_2|} \\
        &\leq \frac{|S_1||S_2|\text{deg}_{max}(S_1)+|S_1||S_2|\text{deg}_{max}(S_2)}{2|S_1||S_2|} \\
        &= \frac12\big(\text{deg}_{max}(S_1)+\text{deg}_{max}(S_2)\big) \\
        &\leq \frac12\big(\text{deg}_G(v_1)+\text{deg}_G(v_2)\big).
    \end{align*} 
    Since this holds for any ordered quasi-bipartition, in particular it holds for one attaining the maximum in Theorem~\ref{theo:min_over_oqbp}, thus
    \[B(G)\leq\frac12\big(\text{deg}_G(v_1)+\text{deg}_G(v_2)\big).\]
    
    On the other hand, if we choose $(\{v_1\},\{v_2\})\in\Gamma$, then 
    \[\frac12\big(\xi(\{v_1\})+\xi(\{v_2\})\big) = \frac12\big(|\delta(\{v_1\})|+|\delta(\{v_2\})|\big)=\frac12\big(\text{deg}_G(v_1)+\text{deg}_G(v_2)\big).\]
    Therefore
    \[B(G)\geq\frac12\big(\text{deg}_G(v_1)+\text{deg}_G(v_2)\big),\]
    and the theorem follows.
\end{proof}

    Note that the $\ell_1$-Fiedler vector $\bm x\in\mathcal{G}$ corresponding to the quasi-bipartition $(\{v_1\},\{v_2\})$ is given by    
    \[x_v=\left\{\begin{matrix*}[r] 1/2, & \text{if } v=v_1,~\\ -1/2, & \text{if } v=v_2,~\\0, & \text{otherwise.}\end{matrix*}\right.\]
    Vectors of this form, with exactly one component equal to $1/2$, exactly one equal to $-1/2$, and all other components equal to zero, will play a key role in the geometric interpretation developed in the next section.
    
    The combinatorial characterization obtained in this section admits a natural geometric reinterpretation. In the next section we study the feasible set as a polyhedral set embedded in the ambient space $\mathbf{R}^V$, and place the finite set $\mathcal{G}$ within that structure. This point of view provides a different perspective on the optimization problems defining $b(G)$ and $B(G)$, and also makes it possible to address structural questions about the corresponding $\ell_1$-Fiedler vectors, such as their location inside the feasible set and how many there are.

\section{The geometry of combinatorial Fiedler theory}\label{sec:geometry}

    Let $\mathcal{H}$ be the hyperplane containing the origin and normal to the all-ones vector $\bm 1$ in $\mathbf{R}^V$, that is
    \[\mathcal{H}=\left\{\bm x\in\mathbf{R}^V~:~\bm 1^\top\bm x=0\right\}=\left\{\bm x\in\mathbf{R}^V~:~\sum_{v\in V}x_v=0\right\}.\]
    Let $\mathcal{B}$ be the $\ell_1$ unit ball centered at the origin in $\mathbf{R}^V$, that is 
    \[\mathcal{B}=\left\{\bm x\in\mathbf{R}^V~:~\|\bm x\|_1\leq1\right\}=\left\{\bm x\in\mathbf{R}^V~:~\sum_{v\in V}|x_v|\leq1\right\}.\]
    The set $\mathcal{B}$ is the convex hull of $\{\pm\bm e_v\}_{v\in V}$ where $\bm e_v$ are the standard basis vectors of $\mathbf{R}^V$; it is called the cross polytope. Now, the set  $\bar{\mathcal{F}}=\mathcal{H}\cap\mathcal{B}=\{\bm x\in\mathbf{R}^V~:~{\sum_{v\in V}x_v=0} ~\wedge~ {\sum_{v\in V}|x_v|\leq1}\}$ is not quite our feasible set, but it can be cleanly characterized in terms of vectors $\bm e_v$ in the following way, where we denote the convex hull of a set $\mathcal{X}\subset\mathbf{R}^V$ by $CH(\mathcal{X})$.
    
\begin{lemma}\label{lemm:barF=CH}
    $\bar{\mathcal{F}}=CH\Big(\left\{\frac12(\bm e_u-\bm e_v)~:~ u\neq v\right\}\Big)$.
\end{lemma}

\begin{proof}
    Let $u$ and $v$ be any two distinct elements of $V$ and define $\bm e_{uv}=\frac12(\bm e_u-\bm e_v)$. Then
    \[\bm1^\top\bm e_{uv}=\frac12\sum_{w\in V}(\bm e_u-\bm e_v)_w=\frac12(1-1)=0~~\Longrightarrow~~\bm e_{uv}\in\mathcal H,\]
    \[\|\bm e_{uv}\|_1=\frac12\sum_{w\in V}\left|\left(\bm e_u-\bm e_v\right)_w\right|=\frac12(1+1)=1~~\Longrightarrow~~\bm e_{uv}\in\mathcal B.\]
    Hence, each generating vector $\bm e_{uv}=\frac12(\bm e_u-\bm e_v)\in\bar{\mathcal F}$ for each pair $u\neq v$. Then, since $\bar{\mathcal F}$ is itself convex, as it is the intersection of two convex sets, we have
    \[CH\big(\left\{\bm e_{uv}~:~ u\neq v\right\}\big) \subset\bar{\mathcal{F}}.\]
    To prove the other inclusion let $\bm x\in\bar{\mathcal F}$ and define $\bm x^+,\bm x^-\in\mathbf{R}^V$ by
\begin{equation}\label{eq:x^+_x^-}
    x^+_v=\left\{\begin{matrix}x_v, & \text{ if } x_v>0\\0, & \text{ if } x_v\leq0\end{matrix}\right.,~~\text{ and }~~x^-_v=\left\{\begin{matrix}0, & \text{ if } x_v\geq0\\|x_v|, & \text{ if } x_v<0\end{matrix}\right..        
\end{equation}
    With this $\bm x=(\bm x^+-\bm x^-)$, and since $\bm x\in\mathcal{H}$,
    \[0=\sum_{v\in V}x_v=\sum_{v\in V}(x_v^+-x_v^-)=\sum_{v\in V}x_v^+-\sum_{v\in V}x_v^-~~\Longrightarrow~~\sum_{v\in V}x_v^+=\sum_{v\in V}x_v^-=:s\in\mathbf{R}_{\geq0}.\]
    And from $\bm x\in\mathcal{B}$ we have
    \[1\geq\sum_{v\in V}|x_v|=\sum_{v\in V}x_v^++\sum_{v\in V}x_v^-=2s~~\Longrightarrow~~s\leq\frac12.\]
    Note that $s=0$ if and only if $\bm x=\bm 0\in\mathbf{R}^V$, in that case the desired inclusion is immediate. If $s\neq 0$, then we can decompose $\bm x$ as a linear combination of the vectors $\bm e_{uv}$ as follows
    \begin{align*}
        \bm x=\bm x^+-\bm x^-   &= \sum_{u\in V}x_u^+\bm e_u-\sum_{v\in V}x_v^-\bm e_v \\
        &= \sum_{u\in V}\left(\frac1s\sum_{v\in V}x_v^-\right)x_u^+\bm e_u-\sum_{v\in V}\left(\frac1s\sum_{u\in V}x_u^+\right)x_v^-\bm e_v \\
        &= \sum_{u,v\in V}\frac1sx_u^+x_v^-(\bm e_u-\bm e_v)\\
        &= \sum_{u,v\in V}\frac2sx_u^+x_v^-\bm e_{uv} = \sum_{\substack{u,v\in V\\u\neq v}}\frac2sx_u^+x_v^-\bm e_{uv}.
    \end{align*}
    Finally, summing the coefficients
    \[\sum_{\substack{u,v\in V\\u\neq v}}\frac2sx_u^+x_v^-=\frac2s\left(\sum_{u\in V}x_u^+\right)\left(\sum_{v\in V}x_v^-\right)=2s\leq1;\]
    we conclude that $\bm x$ is a convex combination of $\left\{\bm e_{uv}~:~u\neq v\right\}\cup\{\bm 0\}$. Since $\bm 0$ is clearly in both sets we have $\bar{\mathcal F}\subset CH\left(\left\{\bm e_{uv}~:~ u\neq v\right\}\right)$ and we are done.
\end{proof}

    Our feasible set $\mathcal{F}$ is the boundary relative to the hyperplane $\mathcal{H}$ of the full polytope $\bar{\mathcal{F}}$, we write $\mathcal{F}=\partial_\mathcal{H}\bar{\mathcal{F}}$ to denote this. The following subsection gives visual depictions of the relevant sets for low-dimensional cases, along with some specific examples.

\subsection{The geometry of $\mathcal{F}$ in low dimensions}
    Consider the case of a vertex set of only three elements $V=\{1,2,3\}$ so that $\mathbf{R}^V\simeq\mathbf{R}^3$. There, the $\ell_1$ unit ball $\mathcal{B}$ is the octahedron represented in the leftmost panel of Figure~\ref{fig:BFGsetsR3} as an orthogonal projection on
    to the plane $\mathcal{H}$. The intersection $\bar{\mathcal{F}}=\mathcal{B}\cap\mathcal{H}$ is the hexagon that is shown in the center panel of the same Figure~\ref{fig:BFGsetsR3}. The one-dimensional feasible set $\mathcal{F}$ is the boundary of this hexagon. The $\mathbf{R}^V$ coordinates of some of the twelve points of $\mathcal{G}\subset\mathcal{F}$ are shown for reference in the rightmost panel of Figure~\ref{fig:BFGsetsR3}.

\begin{figure}[h]
    \centering
    \includegraphics[scale=1.5, page=1]{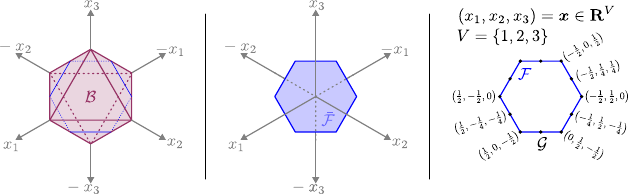}
    \caption{Feasible set $\mathcal{F}\subset\mathbf{R}^V$ and its subset $\mathcal{G}\subset\mathcal{F}$ for $V=\{1,2,3\}$.}
    \label{fig:BFGsetsR3}
\end{figure}
    
    As a first concrete example, Figure~\ref{fig:Ex_Path_3} shows, for the path on three vertices, how the points of $\mathcal{G}$ can be interpreted as one-dimensional drawings of the graph and how the objective function corresponds to the sum of the lengths of all edges in this drawing. 
    
    The leftmost side of Figure~\ref{fig:Ex_Path_3} shows the actual drawings above the line in $\mathbf{R}^1$ for all the points $\bm{g}\in\mathcal{G}$. The coordinates of the $\bm{g}$ vectors are listed in the table together with the corresponding objective function value. A notational convention is adopted for the points of $\mathcal{G}$: the point corresponding to the quasi-bipartition $(P,N)=(\{1,2\},\{3\})$ is denoted by $\bm{g}_{\scriptscriptstyle P,N}=\bm{g}_{\scriptscriptstyle 12,3}$. To the top-right of the figure all points of $\mathcal{G}$ are placed above the hexagon $\mathcal{F}$ in the same positions as in Figure~\ref{fig:BFGsetsR3}. Below that, again in the same order, the value of $f(\bm x)=|x_1-x_2|+|x_2-x_3|$ is depicted schematically as \emph{the distance to the hexagon}; each consecutive gray concentrical hexagon corresponds to an increase of $1/4$ in the $f$ value. This picture already illustrates, in the simplest non-trivial case, how the optimization problems are encoded by the geometry of the feasible set.

\begin{figure}[h]
    \centering
    \includegraphics[scale=1.5, page=2]{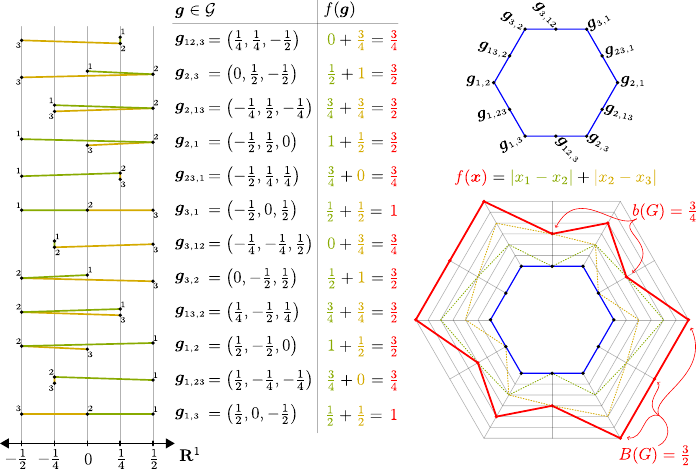}
    \caption{Worked example for the three vertex path graph $G=(V,E)$ with $V=\{1,2,3\}$ and $E=\{12,23\}$.}
    \label{fig:Ex_Path_3}
\end{figure}

    Consider now the case of a vertex set of four elements $V=\{1,2,3,4\}$ such that $\mathbf{R}^V\simeq\mathbf{R}^4$. In this case, the polytope $\bar{\mathcal{F}}=\mathcal{H}\cap\mathcal{B}$ is a cuboctahedron, and the feasible set $\mathcal{F}=\partial_\mathcal{H}\bar{\mathcal{F}}$ is its boundary relative to the three-dimensional hyperplane $\mathcal{H}$. Figure~\ref{fig:BFGsetsR4} shows this configuration by means of an orthogonal projection of the ambient space $\mathbf{R}^4$ onto $\mathcal{H}$, so that the geometry of $\mathcal{F}$ and the finite set $\mathcal{G}\subset\mathcal{F}$ can be visualized in three dimensions. The visible vertices $\bm e_{uv}=\frac12(\bm e_u-\bm e_v)$ are labeled in the right panel in accordance with the axes shown in the left panel; they correspond to the points $\bm{g}_{\scriptscriptstyle u,v}$ with $u\neq v$ in $\mathcal{G}$, those are the quasi-bipartitions $(\{u\},\{v\})$ for $u\neq v$ under the bijection of Proposition~\ref{prop:sum_xis}. The other visible points of $\mathcal{G}$ are also marked as black dots, they lie at the centers of edges and faces of $\mathcal{F}$.
    
\begin{figure}[h]
    \centering
    \includegraphics[scale=1, page=3]{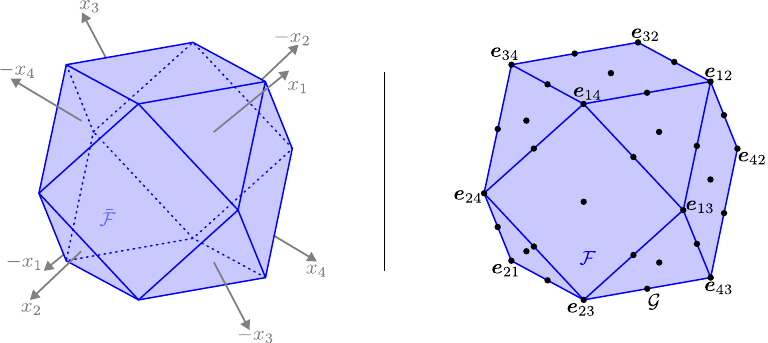}
    \caption{Feasible set $\mathcal{F}\subset\mathbf{R}^V$ and its subset $\mathcal{G}\subset\mathcal{F}$ for $V=\{1,2,3,4\}$.}
    \label{fig:BFGsetsR4}
\end{figure}

    Figure~\ref{fig:Ex_Graph_4.pdf} considers a particular graph as an example of this case of $|V|=4$. A subset of \emph{connected} points of $\mathcal{G}$ are arbitrarily selected and highlighted over the cuboctahedron at the bottom-right of the figure. The corresponding $\mathbf{R}^1$ drawing for those points are shown in the left side of the figure and to their side their corresponding objective value is given. Note that the configuration at the bottom of the pile of drawings $\bm{g}_{\scriptscriptstyle 123,4}=(\frac16,\frac16,\frac16,-\frac12)$, corresponds to a minimizing solution for this particular graph, while the one at the top of the pile $\bm{g}_{\scriptscriptstyle 3,2}=(0,-\frac12,\frac12,0)$ is a maximizing solution. The corresponding quasi-bipartitions for these two points are depicted over the graph at the top-right of the figure.

\begin{figure}[h]
    \centering
    \includegraphics[scale=1.5, page=4]{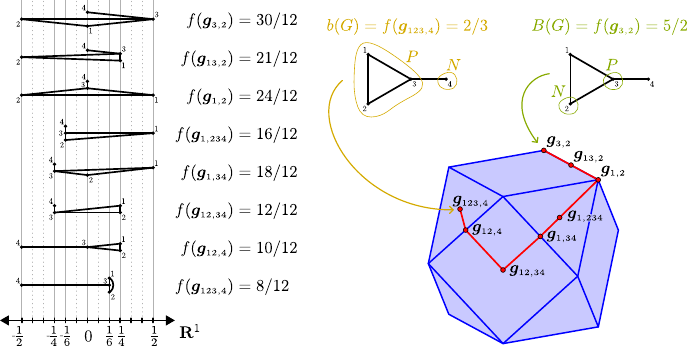}
    \caption{Example for the graph $G=(V,E)$ with $V=\{1,2,3,4\}$ and $E=\{12,13,23,34\}$.}
    \label{fig:Ex_Graph_4.pdf}
\end{figure}

    The placement of the points of $\mathcal{G}$ inside the set $\mathcal{F}$ in the general case, together with a formal description of the \emph{connection} between points of $\mathcal{G}$ and how the value of $f$ changes when moving from point to point is discussed in the next subsection.

\subsection{Geometric interpretation of the optimization problems}

    In the general case, for $|V|=n>4$ the set $\bar{\mathcal{F}}$ is an $(n-1)$-dimensional polytope embedded in the ambient space $\mathbf{R}^V\simeq\mathbf{R}^n$. The feasible set is its shell $\mathcal{F}=\partial_\mathcal{H}\bar{\mathcal{F}}$  composed of $(n-2)$-dimensional faces whose relative interior points have no coordinate equal to zero. Then come $(n-3)$-dimensional faces in whose relative interiors points have exactly one coordinate equal to zero, $(n-4)$-dimensional faces in whose relative interiors points have exactly two coordinates equal to zero, and so on. This ends at the vertices $\bm{e}_{uv}=\bm{g}_{\scriptscriptstyle u,v}=\frac12(\bm e_u-\bm e_v)$ for $u\neq v$ in $V$ ($0$-dimensional faces) for which exactly $n-2$ coordinates are zero.

    Alternative proofs of the results of Section~\ref{sec:intro} can be obtained from this geometric interpretation of the problem. It is convenient to introduce a decomposition of $\mathcal{F}$ into $(n-2)$-dimensional \emph{cells} in whose relative interiors both the weak order between coordinates and the weak signs of each coordinate are preserved. Note that $\mathcal{F}$ (also $\mathcal{G}$) depends only on the vertex set $V$ and not on the edge set $E$ that is considered. The key idea is that, with such a decomposition, the function $f$ of any edge set $E$ is linear in each cell.

    For each bijection $\pi:\{1,\ldots,n\}\rightarrow V$ assigning an order for the coordinates and each number $k$ of non-negative coordinates $k=1,\ldots,n-1$, we define a cell to be the following closed set:
\begin{equation}\label{eq:cell_def}
    \mathcal{C}(\pi,k)=\{\bm x\in\mathcal{F} ~:~ x_{\pi(1)} \geq\cdots\geq x_{\pi(k)}\geq0\geq x_{\pi(k+1)}\geq\cdots\geq x_{\pi(n)}\}.
\end{equation}
    Note that in general $\mathcal{F}$ is decomposed into $(n-1)\cdot n!$ cells; the $12$ cells of the hexagon and some of the $72$ cells of the cuboctahedron are shown in Figure~\ref{fig:Ex_Cells_34}.

\begin{figure}[h]
    \centering
    \includegraphics[scale=1.25, page=5]{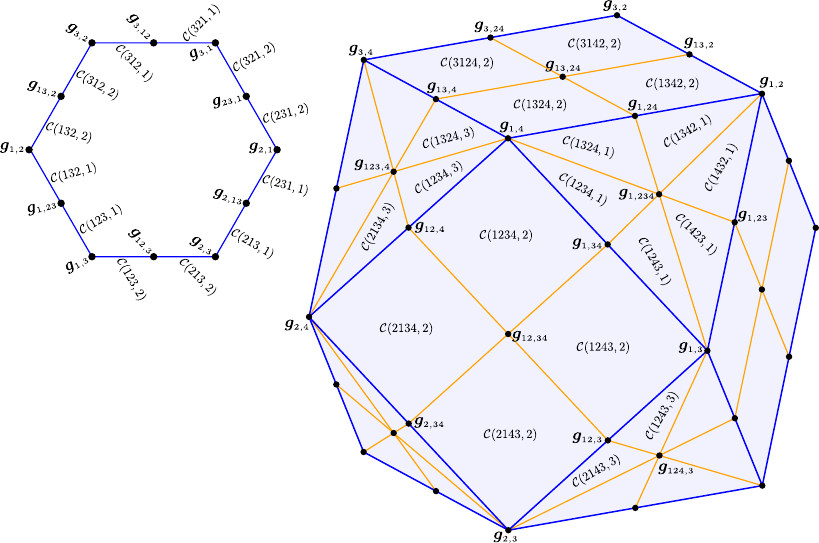}
    \caption{Cell decomposition of the feasible sets for the low-dimensional $|V|=3$ and $|V|=4$ cases.}
    \label{fig:Ex_Cells_34}
\end{figure}

    The interiors (relative to $\mathcal{F}$) of the cells $\mathcal{C}(\pi,k)$  are pairwise disjoint; the union of all the cells equals $\mathcal{F}$. Two different cells might have common lower-dimensional faces. In particular, each cell $\mathcal{C}(\pi,k)$ has $k(n-k)$ vertices which are obtained precisely when all positive coordinates are equal, all negative coordinates are equal, and the remaining coordinates are zero. This fact follows from the simplex decomposition of the cells given later in Lemma~\ref{lemm:C=C^+_C^-}, we state it here as a lemma for reference.

\begin{lemma}
    The set of all vertices of the cells $\mathcal{C}(\pi,k)$ coincides with $\mathcal{G}$.
\end{lemma}
    
    On each cell $\mathcal{C}(\pi,k)$, the weak order of coordinates is fixed, so for each $uv\in E$ the corresponding difference $x_u-x_v$ is either always non-negative or always non-positive. Hence, the corresponding term $|x_u-x_v|$ of $f$ can be replaced either by $x_u-x_v$ or by $-(x_u-x_v)$, so we have the following result.

\begin{lemma}
    For every graph $G=(V,E)$ and every cell $\mathcal{C}(\pi,k)$, the restriction of $f$ to $\mathcal{C}(\pi,k)$ is linear.
\end{lemma}

    With these last results, we have the following alternative proofs for Lemma~\ref{lemm:same_pos_same_neg} and Theorem~\ref{theo:two_highest_geg}.

\begin{proof} \emph{(of Lemma~\ref{lemm:same_pos_same_neg})}
    Each one of the cells $\mathcal{C}(\pi,k)$ is a convex polytope, as it is defined by a finite set of linear weak inequalities and is bounded because $\mathcal{F}$ is bounded. A linear function attains its extrema over a convex polytope at the vertices of the polytope. The extrema of $f$ over $\mathcal{F}$ are therefore attained at the vertices of the cells $\mathcal{C}(\pi,k)$; these are precisely the points of $\mathcal{G}$. Hence, the optimization problems defining $b(G)$ and $B(G)$ admit solutions in $\mathcal{G}$.
\end{proof}

\begin{proof} \emph{(of Theorem~\ref{theo:two_highest_geg})}
    The function $f(\bm x)=\sum_{uv\in E}|x_u-x_v|$ is convex on $\mathbf{R}^V$, since it is a finite sum of convex functions. As $\bar{\mathcal F}$ is a convex polytope, a maximum of $f$ over $\bar{\mathcal F}$ is attained at a vertex of $\bar{\mathcal F}$. These vertices are precisely the vectors of the form $\frac12(\bm e_u-\bm e_v)$ with $u\neq v$, and they all belong to $\mathcal F$. Therefore,
    \[B(G)=\max_{\bm x\in\mathcal F}f(\bm x)=\max_{u\neq v}f\left(\frac12(\bm e_u-\bm e_v)\right).\]
    For any fixed pair $u\neq v$, each term $|x_a-x_b|$ in $f\left(\frac12(\bm e_u-\bm e_v)\right)$ is equal to $0$ when the edge $ab\in E$ is not incident to $u$ or $v$, is equal to $1/2$ when $ab$ is incident to exactly one of $u$ and $v$, and is equal to $1$ when $ab=uv$. Hence $f\left(\frac12(\bm e_u-\bm e_v)\right)=\frac12(\deg_G(u)+\deg_G(v))$ and then,
    \[B(G)=\max_{u\neq v}\left\{\frac12(\deg_G(u)+\deg_G(v))\right\}=\frac12\big(\deg_G(v_1)+\deg_G(v_2)\big),\]
    where $v_1$ and $v_2$ are two vertices of the highest degree in $G$.
\end{proof}

    The discussion above places the optimization problems that define $b(G)$ and $B(G)$ in a polyhedral setting. The set $\mathcal{G}$ consists of the vertices of the closed cells of $\mathcal{F}$, and the function $f$ is linear on each cell. In the next subsection we examine more closely how the finite set $\mathcal{G}$ sits inside $\mathcal{F}$, with particular attention to the way its points are connected through the surrounding cell structure.
    
\subsection{The structure of $\mathcal{G}$ inside $\mathcal{F}$}

    Recall the notation introduced in Proposition~\ref{prop:sum_xis}; for each ordered quasi-bipartition $(P,N)\in\Gamma$, we denote its corresponding point by $\bm g_{\scriptscriptstyle P,N}\in\mathcal{G}$. As we saw in the last subsection, the set $\mathcal{G}$ coincides with the set of vertices of the cells of $\mathcal{F}$.

    Among the points of $\mathcal{G}$, there are two types that play a distinguished role. On the one hand, there are the points with only two non-zero coordinates, namely the points $\bm g_{\scriptscriptstyle \{u\},\{v\}}\in\mathcal{G}$ for $u\neq v$ in $V$. These are the vertices of the polytope $\bar{\mathcal{F}}$, and the maximum of $f$ is attained at points of this form. We call these points the \textbf{extremes} of $\mathcal{G}$. On the other hand, there are the points of $\mathcal{G}$ with no zero coordinates, that is, the points $\bm g_{\scriptscriptstyle P,N}\in\mathcal{G}$ such that $P\cup N=V$, with $P\cap N=\varnothing$. We call these points \textbf{centers} of $\mathcal{G}$.

    As an example consider the cell $\mathcal{C}(1243,2)$ in Figure~\ref{fig:Ex_Cells_34} above. It has four vertices, but among them just one center $\bm{g}_{\scriptscriptstyle 12,34}=\left(\frac14,\frac14,-\frac14,-\frac14\right)$ and one extreme $\bm{g}_{\scriptscriptstyle 1,3}=\left(\frac12,0,-\frac12,0\right)$. In general, for each cell $\mathcal{C}(\pi,k)$, there is exactly one center and one extreme among its vertices $\bm g_{\scriptscriptstyle P,N}\in\mathcal{G}\cap\mathcal{C}(\pi,k)$. They are the points with $P=\{\pi(1),\dots,\pi(k)\},~N=\{\pi(k+1),\dots,\pi(n)\}$ and $P=\{\pi(1)\},~N=\{\pi(n)\}$, respectively.

    This structure suggests considering not only the set $\mathcal{G}$ itself, but also the neighboring relation induced on it by the cell decomposition of $\mathcal{F}$. We define the \textbf{quasi-bipartition graph} as the graph whose vertex set is $\mathcal{G}$, where two points of $\mathcal{G}$ are adjacent whenever they are joined by an edge of some cell of $\mathcal{F}$.

    In order to understand which elements of $\mathcal{G}$ are adjacent in the quasi-bipartition graph, we leverage the fact that each cell can be decomposed into a \emph{positive part} and a \emph{negative part}, similar to what was done in the proof of Lemma~\ref{lemm:barF=CH}. 

\begin{lemma}\label{lemm:C=C^+_C^-}
    Each cell $\mathcal{C}(\pi,k)$ can be decomposed into the Minkowski sum $\mathcal{C}(\pi,k)=\mathcal{C}^+(\pi,k)+\mathcal{C}^-(\pi,k)$, where $\mathcal{C}^+(\pi,k)=\{\bm{x}^+~:~\bm{x}\in\mathcal{C}(\pi,k)\}$ and $\mathcal{C}^-(\pi,k)=\{-\bm{x}^-~:~\bm{x}\in\mathcal{C}(\pi,k)\}$, with $\bm{x}^+$ and $\bm{x}^-$ defined as in Equation~\eqref{eq:x^+_x^-}:
    \[x^+_v=\left\{\begin{matrix}x_v, & \text{ if } x_v>0\\0, & \text{ if } x_v\leq0\end{matrix}\right.,~~\text{ and }~~x^-_v=\left\{\begin{matrix}0, & \text{ if } x_v\geq0\\|x_v|, & \text{ if } x_v<0\end{matrix}\right..\]
    Furthermore, these parts $\mathcal{C}^+(\pi,k)$ and $\mathcal{C}^-(\pi,k)$ are, respectively, a $(k-1)$-simplex and a $(n-k-1)$-simplex.
\end{lemma}

\begin{proof}
    The set equality $\mathcal{C}(\pi,k)=\mathcal{C}^+(\pi,k)+\mathcal{C}^-(\pi,k)$ its true, since each $\bm{x}\in\mathcal{C}(\pi,k)$ can be rewritten as $\bm x=\bm x^+-\bm x^-$. To prove the second part, we present the affine independent vertices of the simplices explicitly:
    \begin{equation}\label{eq:C+=CH}
        \mathcal{C}^+(\pi,k)=CH\big(\{\bm{p}_1,\ldots,\bm{p}_k\}\big),~~\text{where}~~\bm{p}_i=\frac1{2i}\sum_{j=1}^i\bm{e}_{\pi(j)},
    \end{equation}
    and
    \begin{equation}\label{eq:C-=CH}
        \mathcal{C}^-(\pi,k)=CH\big(\{\bm{q}_{1},\ldots,\bm{q}_{n-k}\}\big),~~\text{where}~~\bm{q}_i=-\frac1{2i}\sum_{j=1}^i\bm{e}_{\pi(n-j+1)}.    
    \end{equation}
   We prove Equation~\eqref{eq:C+=CH}, the argument for  Equation~\eqref{eq:C-=CH} is similar. Let $\bm y\in\mathcal C^+(\pi,k)$ so that the sum  $\sum_{j=1}^ky_{\pi(j)}=1/2$, and $y_{\pi(k+1)}=0$. We have\footnote{The algebraic trick, used again latter, is $\displaystyle\sum_{j=1}^{k}\sum_{i=j}^{k}a_{i,j}=\sum_{1\leq j\leq i\leq k}a_{i,j}=\sum_{i=1}^{k}\sum_{j=1}^{i}a_{i,j}$, see Equation (2.32) in~\cite{graham1994concrete}.}
    \begin{align*}
        \bm y
            &=\sum_{j=1}^k y_{\pi(j)}\bm e_{\pi(j)}
            =\sum_{j=1}^k\left(\sum_{i=j}^k\big(y_{\pi(i)}-y_{\pi(i+1)}\big)\right)\bm e_{\pi(j)}\\
            &=\sum_{i=1}^k\big(y_{\pi(i)}-y_{\pi(i+1)}\big)\sum_{j=1}^i \bm e_{\pi(j)}
            =\sum_{i=1}^k 2i\big(y_{\pi(i)}-y_{\pi(i+1)}\big)\frac1{2i}\sum_{j=1}^i \bm
            e_{\pi(j)}
            =\sum_{i=1}^k \lambda_i \bm p_i,
    \end{align*}

    where $\lambda_i:=2i\big(y_{\pi(i)}-y_{\pi(i+1)}\big)$ for $i=1,\ldots,k$. Moreover,
    \[\sum_{i=1}^k \lambda_i=\sum_{i=1}^k 2i(y_{\pi(i)}-y_{\pi(i+1)})=2\sum_{i=1}^k y_{\pi(i)}=1,\]
    so $\mathcal{C}^+(\pi,k)\subset CH\big(\{\bm{p}_1,\ldots,\bm{p}_k\}\big)$.

    To see the other inclusion, let $\bm y\in CH\big(\{\bm{p}_1,\ldots,\bm{p}_k\}\big)$. That is to say
    \[\bm y=\sum_{i=1}^k \lambda_i \bm p_i ~~\text{with}~~\lambda_i\geq 0 ~~\text{and}~~ \sum_{i=1}^k \lambda_i=1.\]
    Now from
    \[\bm y=\sum_{i=1}^k \lambda_i \bm p_i = \sum_{i=1}^k\frac{\lambda_i}{2i}\sum_{j=1}^i\bm{e}_{\pi(j)} = \sum_{i=1}^k\sum_{j=1}^i\frac{\lambda_i}{2i}\bm{e}_{\pi(j)} = \sum_{j=1}^k\left(\sum_{i=j}^k\frac{\lambda_i}{2i}\right)\bm{e}_{\pi(j)},\]
    we see that for each $j=1,\ldots,k$ the corresponding $\bm y$ component is 
    $y_{\pi(j)} = \sum_{i=j}^k\lambda_i/2i$, so 
    \[y_{\pi(1)}\geq y_{\pi(2)}\geq\cdots\geq y_{\pi(k)}\geq 0.\]
    Finally,
    \[\sum_{j=1}^k y_{\pi(j)} = \sum_{j=1}^k \sum_{i=j}^k \frac{\lambda_i}{2i} = \sum_{i=1}^k\sum_{j=1}^i \frac{\lambda_i}{2i} = \sum_{i=1}^k \frac{\lambda_i}{2} = \frac12.\]
    Therefore, $\bm y\in\mathcal C^+(\pi,k)$, and thus $CH\big(\{\bm p_1,\ldots,\bm p_k\}\big)\subset\mathcal C^+(\pi,k)$ completing the proof.
    \end{proof}

    For a concrete example consider $V=\{1,2,3,4,5\}$, and take the cell $\mathcal{C}(13245,2)$ of the feasible set $\mathcal{F}$. Then 
    \[\mathcal{C}^+(13245,2)=\textstyle CH(\{(\frac12,0,0,0,0),~(\frac14,0,\frac14,0,0)\}),\]
    \[\mathcal{C}^-(13245,2)=\textstyle CH(\{(0,0,0,0,-\frac12),~(0,0,0,-\frac14,-\frac14),~(0,-\frac16,0,-\frac16,-\frac16)\}).\]

    We refer to Fukuda~\cite{Fukuda2004} for a general treatment of polytopes constructed as Minkowski sums. In the present situation, however, the geometry is particularly simple: a coordinate that may be nonzero on $\mathcal C^+(\pi,k)$ is identically zero on $\mathcal C^-(\pi,k)$, and conversely.  Hence every point of $\mathcal C(\pi,k)$ has a unique representation as $\bm{x}^++\bm{x}^-$ with $\bm{x}^+\in\mathcal C^+(\pi,k)$ and $\bm{x}^-\in\mathcal C^-(\pi,k)$. Equivalently, the addition map $(\bm{x}^+,\bm{x}^-)\mapsto \bm{x}^++\bm{x}^-$ is an affine isomorphism from $\mathcal C^+(\pi,k)\times\mathcal C^-(\pi,k)$ to $\mathcal C(\pi,k)$. Consequently, every edge of $\mathcal C(\pi,k)$ is obtained by fixing one factor at a vertex and taking an edge of the other. % \cite{pfeifle2012polytopality}. % SEE SECTION 2 THERE 
    Translating this description into the language of ordered quasi-bipartitions gives the desired characterization of adjacency in the quasi-bipartition graph. 
    
    In terms of ordered quasi-bipartitions, the vertices of $\mathcal{C}^+(\pi,k)$ correspond to the nested positive sets
    \[\{\pi(1)\}\varsubsetneq \{\pi(1),\pi(2)\}\varsubsetneq\cdots\varsubsetneq\{\pi(1),\ldots,\pi(k)\},\]
    while the vertices of $\mathcal{C}^-(\pi,k)$ correspond to the nested negative sets
    \[\{\pi(n)\}\varsubsetneq\{\pi(n-1),\pi(n)\}\varsubsetneq\cdots\varsubsetneq\{\pi(k+1),\ldots,\pi(n)\}.\]
    Therefore, moving along an edge of $\mathcal{C}(\pi,k)$ means keeping one of the two sets fixed and replacing the other by another member of one of these nested chains.

\begin{proposition}\label{prop:adjacency_quasibipartitions}
    Two distinct points of $\mathcal{G}$ are adjacent in the quasi-bipartition graph if and only if one of the two sets is the same for their corresponding ordered quasi-bipartitions, while the other two are comparable by inclusion. Formally,
    \[\bm g_{\scriptscriptstyle P,N}\sim\bm g_{\scriptscriptstyle P',N'}~~\Longleftrightarrow~~
    \Big((P=P')\wedge(N\varsubsetneq N'\vee N\varsupsetneq N')\Big)\vee\Big((N=N')\wedge(P\varsubsetneq P'\vee P\varsupsetneq P')\Big).\]
\end{proposition}

\begin{proof} 
    The forward implication follows immediately from the above discussion. Conversely, if two ordered quasi-bipartitions have one set equal and the other two comparable by inclusion, one can choose an order $\pi$ such that the corresponding points lie in a common cell and differ in exactly one simplex factor. Therefore, they are adjacent. 
\end{proof}

    The quasi-bipartition graph provides a convenient discrete model for the behavior of $f$ on $\mathcal{F}$. It is useful because it retains the adjacency structure of the cell decomposition on the finite set $\mathcal{G}$. In this way, questions about the local behavior of $f$ on $\mathcal{F}$ can be reduced to questions about how $f$ varies along the edges of this graph. Since $f$ is affine on each cell, this leads naturally to studying its directional differences along those edges, which turn out to have a clean combinatorial characterization.

    If $\bm g\sim\bm g'$, we denote by $\Delta(\bm g\to\bm g')f$ the directional difference of $f$ along the oriented edge from $\bm g$ to $\bm g'$. Then, depending on which of the two parts changes, we have:
    \begin{equation}\label{eq:Delta(gg)f}\begin{aligned}
        \Delta(\bm g_{\scriptscriptstyle P,N}\to\bm g_{\scriptscriptstyle P',N})f &= \frac12\big(\xi(P')-\xi(P)\big), \\
        \Delta(\bm g_{\scriptscriptstyle P,N}\to\bm g_{\scriptscriptstyle P,N'})f &= \frac12\big(\xi(N')-\xi(N)\big). 
    \end{aligned}\end{equation}
        
    Now, in particular, let $\bm g_{\scriptscriptstyle P,N}\in\mathcal{G}$ be a non-center, and write $Z=V\setminus(P\cup N)$. Its two \textbf{centralizing directions} are the oriented edges $\bm g_{\scriptscriptstyle P,N}\to \bm g_{\scriptscriptstyle P\cup Z,N}$ and $\bm g_{\scriptscriptstyle P,N}\to \bm g_{\scriptscriptstyle P,N\cup Z}$.

\begin{lemma}\label{lemm:centralizing_direction_negative}
    Let $G$ be connected and let $\bm g_{\scriptscriptstyle P,N}\in\mathcal{G}$ be a non-center. Then at least one of the two centralizing directions has negative difference:
    \[\Delta(\bm g_{\scriptscriptstyle P,N}\to \bm g_{\scriptscriptstyle P\cup Z,N})f<0 ~~\text{or}~~ \Delta(\bm g_{\scriptscriptstyle P,N}\to \bm g_{\scriptscriptstyle P,N\cup Z})f<0.\]
\end{lemma}

\begin{proof}
    Suppose, for the sake of contradiction, that both directions are nonnegative,
    \[\Delta(\bm g_{\scriptscriptstyle P,N}\to \bm g_{\scriptscriptstyle P\cup Z,N})f\geq0 ~~\text{and}~~ \Delta(\bm g_{\scriptscriptstyle P,N}\to \bm g_{\scriptscriptstyle P,N\cup Z})f\geq0.\]
    By Equation~\eqref{eq:Delta(gg)f} this is equivalent to
    \[\frac12\big(\xi(P\cup Z)-\xi(P)\big)\geq0 ~~\text{and}~~ \frac12\big(\xi(N\cup Z)-\xi(N)\big)\geq0.\]
    Recalling the definition $\xi(S)=|\delta(S)|/|S|$, we have
    \[\frac{|\delta(P\cup Z)|}{|P|+|Z|}-\frac{|\delta(P)|}{|P|}\geq0 ~~\text{and}~~ \frac{|\delta(N\cup Z)|}{|N|+|Z|}-\frac{|\delta(N)|}{|N|}\geq0.\]
    Since $\delta(S)=\delta(S^c)$, this becomes
    \[\frac{|\delta(N)|}{|P|+|Z|}-\frac{|\delta(P)|}{|P|}\geq0 ~~\text{and}~~ \frac{|\delta(P)|}{|N|+|Z|}-\frac{|\delta(N)|}{|N|}\geq0.\]
    Rearranging, we get
    \[|\delta(N)|-|\delta(P)|\geq\frac{|Z|}{|P|}|\delta(P)| ~~\text{and}~~ |\delta(P)|-|\delta(N)|\geq\frac{|Z|}{|N|}|\delta(N)|,\]
    hence
    \[|\delta(N)|-|\delta(P)|\geq|Z|\xi(P) ~~\text{and}~~ |\delta(P)|-|\delta(N)|\geq|Z|\xi(N).\]
    Since $\bm g_{\scriptscriptstyle P,N}\in\mathcal{G}$ is not a center, we have $Z\neq\varnothing$. Since $G$ is connected and both $P$ and $N$ are nonempty proper subsets of $V$, we also have $\xi(P),\xi(N)>0$. Therefore
    \[|\delta(N)|-|\delta(P)|>0 ~~\text{and}~~ |\delta(P)|-|\delta(N)|>0,\]
    which is impossible.
\end{proof}

\begin{corollary}\label{cor:minimum_on_G_at_centers}
    If $G$ is connected, every minimum point of $f$ on $\mathcal{G}$ is a center.
\end{corollary}

This last result restates, in the present geometric setting, Corollary~3.3 of Andrade and Dahl~\cite{andrade2024combinatorial}, which states that, in the context of connected simple graphs $G$, the $\ell_1$-Fiedler vectors of $b(G)$ contain no zero coordinates. 

A direct consequence of this fact, which will be used in the next section to count the number of $\ell_1$-Fiedler vectors for $b(G)$, is the following:

\begin{theorem}[Andrade and Dahl~\cite{andrade2024combinatorial}]\label{thm:andrade}
For any graph $G=(V,E)$,
\[b(G)=\frac{|V|}{2} \min_S \frac{|\delta{S}|}{|S||V \backslash S|},\]
where the minimum is taken for nonempty subsets $S$ of $V$ such that $S\neq V$ and both $S$ and its complement induce connected subgraphs of $G$.
\end{theorem}

That is, $b(G)$ corresponds to a sparsest cut in $G$; we want to partition $V$ into $S$ and its complement in $V$, so that both of $S$ and $V\backslash S$ are large sets, but with only few edges between $S$ and $V\backslash S.$ The $\ell_1$-Fiedler vector $\bm x^S\in\mathcal{G}$ for $b(G)$ corresponding to a sparsest cut $\delta(S)$ is $\bm{x}^S=(x_v)_{v\in V}$ given by
     \[x_v=\left\{\begin{array}{rl}  \frac{1}{2|S|},  & \text{ if }  v \in S,~\\ -\frac{1}{2|V\backslash S|}, &  \text{ if } v \in V\backslash S.\end{array}\right.\]

\section{Examples}\label{sec:examples}

In this section, we illustrate the results of the previous sections on several graph families. For each example, the goal is to describe the points of $\mathcal{G}$ that realize the optimization problems defining $b(G)$ and $B(G)$ and to explicitly count the corresponding $\ell_1$-Fiedler vectors. This also provides concrete examples for the counting problem studied in the next section. Only $\ell_1$-Fiedler vectors in $\mathcal{G}$ are considered here, as otherwise there would be an infinite number of $\ell_1$-Fiedler vectors in many cases. We omit stating this restriction for each example throughout this section.

Each $\ell_1$-Fiedler vector for $B(G)$ corresponds to an ordered quasi-bipartition $(P,N)$ of $V$, see Proposition~\ref{prop:sum_xis}. If $G$ has at least two vertices of same highest degree, then  each of $P$ and $N$ is a subset of  non-adjacent vertices of highest degree of $G$. If $G$ has only one vertex of highest degree, then one of $P,N$ consists of only this vertex, and the other one is a subset of non-adjacent vertices among the vertices of second-highest degree. This observation allows counting the number of $\ell_1$-Fiedler vectors for $B(G).$

\subsection{The complete graph $K_n$}

The complete graph on $n$ vertices $K_n$ has $B(K_n)= n-1.$ The quasi-bipartition $(P,N)$ corresponding to an $\ell_1$-Fiedler vector for $B(K_n)$ has exactly one vertex from $K_n$ in $P$ and another one in $N$. Then, the number of $\ell_1$-Fiedler vectors for $B(K_n)$ is $n^2-n.$

On the other hand, $K_n$ has $b(K_n)=\frac{n}{2}$, also see~\cite{andrade2024combinatorial}. Any partition of $V$ into two non-empty subsets of vertices gives an  $\ell_1$-Fiedler vector for $b(K_n)$. Then, the number of $\ell_1$-Fiedler vectors for $b(K_n)$ is $2^n-2.$

Note that the number of vertices and facets of the polytope $\bar{\mathcal{F}}$, the $(n-1)$-dimensional cuboctahedron~\cite{doehlert1972experimental}, is precisely $n^2-n$ and $2^n-2$, respectively.  Each vertex of $\bar{\mathcal{F}}$ gives a solution for the minimization problem in Equation~(\ref{eq:B(G)}). Each center gives a solution to the optimization problem in Equation~(\ref{eq:b(G)}); each facet of $\bar{\mathcal{F}}$ contains one center.

\subsection{The wheel graph $W_n$}

Let $n\geq 4.$ The wheel graph on $n$ vertices consists of a cycle on $n-1$ vertices and one additional vertex that is adjacent to each vertex of the cycle. $W_n$ has $B(W_n)=\frac{n+2}{2}.$
In the ordered quasi-bipartition $(P,N)$ corresponding to an $\ell_1$-Fiedler vector for $B(W_n)$, one of $P$ and $N$ consists of the high-degree vertex, and the other one consists of a non-empty subset of pairwise non-adjacent vertices from the cycle on $n-1$ vertices. The number of ways to choose such a subset is given by the Lucas number $L(n-1)$ minus $1$, see~\cite{Prodinger82}. The Lucas numbers $L(n)$ satisfy $L(n)=L(n-1)+L(n-2)$ for $n \geq 2$, and $L(0)=2$, and $L(1)=1.$ 
Then, the number of $\ell_1$-Fiedler vectors for $B(W_n)$ is
\[2\left(  \left( \frac{1+\sqrt{5}}{2}\right)^{n-1}+\left(\frac{1-\sqrt{5}}{2}\right)^{n-1}  \right)-2 .\]
This is twice the Lucas number $L(n-1)$ minus $2$.

In order to determine $b(W_n)$, we observe that each $\ell_1$-Fiedler vector for $b(W_n)$ corresponds to a partition of $V$ into two subsets, where one subset consists of a set of $x$ consecutive vertices on the cycle, for some integer $x.$ Then $b(W_n) \leq \frac{n}{2}\frac{x+2}{x(n-x))}$. We verify that this function is minimized if we take $x=\left\lfloor\sqrt{2n+4}-2\right\rfloor$ or $x=\left\lceil\sqrt{2n+4}-2\right\rceil$. It follows that
\[b(W_n)=\min\left\{\frac{n}{2}\frac{\left\lfloor\sqrt{2n+4}-2\right\rfloor+2}{\left\lfloor\sqrt{2n+4}-2\right\rfloor(n-\left\lfloor\sqrt{2n+4}-2\right\rfloor)}\ \ ,\  \frac{n}{2}\frac{\left\lceil\sqrt{2n+4}-2\right\rceil+2}{\left\lceil\sqrt{2n+4}-2\right\rceil(n-\left\lceil\sqrt{2n+4}-2\right\rceil)} \right\}.\]  

There are $(2n-2)$ $\ell_1$-Fiedler vectors for $b(W_n)$.

We remark that in~\cite{andrade2024combinatorial} the formula $b(W_n)=\frac{n}{n-2}$ is given, which coincides with the formula for $b(W_n)$ given here for $n=4,5,6,7,8$ but gives an incorrect value for $n=9.$ We have $b(W_9)=\frac{5}{4}.$

\subsection{The cycle graph $C_n$}\label{sec:cycle}

Let $n \geq 3.$ The cycle graph on $n$ vertices $C_n$ has $B(C_n)=2.$  Each $\ell_1$-Fiedler vector for $B(C_n)$
corresponds to an ordered quasi-bipartition $(P,N)$ of the vertex set of $C_n$, where both of $P$ and $N$ are non-empty subsets of pairwise non-adjacent vertices. Counting the number of $\ell_1$-Fiedler vectors for $B(C_n)$ then generalizes the problem of counting the number of subsets of pairwise non-adjacent vertices from a cycle  from~\cite{Prodinger82}. Here we  count the number of two  disjoint subsets of pairwise non-adjacent vertices, instead of just one subset. Each such ordered quasi-bipartition $(P,N)$ can be encoded by a word $w$ of length $n$ from the alphabet $\{r,b,0\}$, where no two consecutive $r$ and no two consecutive $b$ appear, and also the first and the last entry are not both $b$ and not both $r$. Also, at least one entry $r$ and one entry $b$ is needed in $w$. The entries in the word with letter $r$ are elements from $P$, entries with letter $b$ are elements from $N$, and elements with letter $0$ get coordinate $0.$ We can model this with a directed graph, whose adjacency matrix $A$ is 
\[A=\begin{blockarray}{cccc}
r & b & 0 \\
\begin{block}{(ccc)c}
  0 & 1 & 1  & r \\
  1 & 0 & 1  & b \\
  1 & 1 & 1  & 0 \\
\end{block}
\end{blockarray}\]
 The number of words $w$, but maybe only using at most two instead of all three letters, is the number of closed walks of length $n$ in this graph. This number of words is equal to the trace of $A^n.$
 The eigenvalues of $A$ are $1-\sqrt{2}$, $1+\sqrt{2}$, and $-1$. Then, there are 
$(1-\sqrt{2})^n+(1+\sqrt{2})^n+(-1)^n$ such words $w.$ Note that we also counted words that do not use both letters $b$ and $r$. Therefore, we need to subtract the number of words that only use at most one of letters $r$ and $b$. By~\cite{Prodinger82}, we subtract twice the Lucas number $L(n)$. Since also the empty subset is counted in $L(n)$, we add $1.$ The number of $\ell_1$-Fiedler vectors for $B(C_n)$ is
\[(1+\sqrt{2})^n +(1-\sqrt{2})^n  - 2\left(  \left( \frac{1+\sqrt{5}}{2}\right)^{n}+\left(\frac{1-\sqrt{5}}{2}\right)^{n}  \right) +(-1)^n +1.\]
This is the difference between the Pell-Lucas number $P\ell(n)$ and twice the Lucas number $L(n)$, plus $(-1)^n$ plus $1$.
The Pell-Lucas numbers $P\ell(n)$ satisfy $P\ell(n)=2P\ell(n-1)+P\ell(n-2)$ for $n\geq 2$, and $P\ell(0)=P\ell(1)=2.$
%The Lucas numbers $L(n)$ satisfy $L(n)=L(n-1)+L(n-2)$ for $n \geq 2$, and $L(0)=2$, and $L(1)=1.$\\

The cycle graph $C(n)$ has $b(C_n)= \frac {n}{ \lceil \frac{n}{2}  \rceil \lfloor \frac{n}{2} \rfloor  }$, see also~\cite{andrade2024combinatorial}. 
For $n$ even, there are $n$ $\ell_1$-Fiedler vectors for $b(C_n)$, and for $n$ odd there are $2n$ $\ell_1$-Fiedler vectors for $b(C_n)$. These correspond to the partition of $V$ into two  subsets of $\lceil \frac{n}{2}  \rceil$ and $\lfloor \frac{n}{2}  \rfloor$ consecutive vertices of $C_n$.

\subsection{The path graph $P_n$}

Let $n \geq 4.$ The path graph on $n$ vertices $P_n$ has $B(P_n)= 2$. 
 Each $\ell_1$-Fiedler vector for $B(P_n)$
corresponds to an ordered quasi-bipartition $(P,N)$ of the vertex set of $P_n$, where both of $P$ and $N$ are non-empty subsets of pairwise non-adjacent vertices. As in Section~\ref{sec:cycle}, each such ordered quasi-bipartition $(P,N)$ can be encoded by a word $w$ of length $n$ from the alphabet $\{r,b,0\}$, where no two consecutive $r$ and no two consecutive $b$ appear. In addition, at least one entry $r$ and one entry $b$ are needed in $w$. But now, the first and last entry of $w$ can be any letter. 
The number of such words $w$, but maybe only using at most two instead of all three letters, is the number of length walks $n-1$ in the graph with adjacency matrix $A$ from Section~\ref{sec:cycle}. This number of words is equal to the sum of the elements of $A^{n-1}$.\\
We define $a(n)=\frac{ (1+\sqrt{2})^n + (1-\sqrt{2})^n +2(-1)^n}{4}$, \ $b(n)=\frac{ (1+\sqrt{2})^n + (1-\sqrt{2})^n -2(-1)^n}{4}$, \ $c(n) = \frac{ (1+\sqrt{2})^n - (1-\sqrt{2})^n}{2\sqrt{2}}$, \ and $d(n)=\frac{ (1+\sqrt{2})^n + (1-\sqrt{2})^n }{2}.$
It follows by induction on $n$, that for $n\geq 1$, the matrix $A^n$ has the form
\[
A^n=\begin{blockarray}{cccc}
r & b & 0 \\
\begin{block}{(ccc)c}
  a(n) & b(n) & c(n)  & r \\
  b(n) & a(n) & c(n)  & b \\
  c(n) & c(n) & d(n)  & 0 \\
\end{block}
\end{blockarray}
 \]
The sum of the number of elements in $A^n$ is then
\[2a(n)+2b(n)+4c(n)+d(n) = \frac{(1+\sqrt{2})^{n+2}+(1-\sqrt{2})^{n+2}}{2}.\]
The number of words $w$, but maybe only using at most two instead of all three letters, then is
\[\frac{(1+\sqrt{2})^{n+1}+(1-\sqrt{2})^{n+1}}{2}.\]
We need to subtract the number of words that use at most one of the letters $r$ and $b$. By~\cite{Prodinger82}, we subtract twice the Fibonacci number 
\[Fib(n+2)=\frac{1}{\sqrt{5}}\left(
\left( \frac{1+\sqrt{5}}{2}\right)^{n+2}-\left( \frac{1-\sqrt{5}}{2}\right)^{n+2}\right),\]
and add $1$ because the word consisting of only zeros is also counted with the Fibonacci number in~\cite{Prodinger82}.

We find that the number of $\ell_1$-Fiedler vectors for $B(P_n)$ is 
\[\frac{(1+\sqrt{2})^{n+1}+(1-\sqrt{2})^{n+1}}{2}-\frac{2}{\sqrt{5}}
\left(\left( \frac{1+\sqrt{5}}{2}\right)^{n+2}-\left( \frac{1-\sqrt{5}}{2}\right)^{n+2}\right)+1.\]
This is exactly the difference between the modified Pell number $Pe(n+1)$ and twice the Fibonacci number $Fib(n+2)$ plus 1.
The modified Pell numbers $Pe(n)$ satisfy $Pe(n)=2Pe(n-1)+Pe(n-2)$, for $n\geq 2$, and $Pe(1)=Pe(0)=1.$ The Fibonacci numbers $Fib(n)$ satisfy $Fib(n)=Fib(n-1)+Fib(n-2)$ for $n\geq 3$, and $Fib(1)=Fib(2)=1.$

The path graph $P(n)$ has
\[b(P_n)=\frac{n}{2} \frac {1}{ \lceil \frac{n}{2}  \rceil \lfloor \frac{n}{2} \rfloor},\] 
also see~\cite{andrade2024combinatorial}. 
For $n$ even, there are two $\ell_1$-Fiedler vectors for $b(P_n)$, and for $n$ odd there are four $\ell_1$-Fiedler vectors for $b(P_n)$. These correspond to the partition of $V$ into two  subsets of $\lceil \frac{n}{2}  \rceil$ and $\lfloor \frac{n}{2}  \rfloor$ consecutive vertices of $P_n$.  

\subsection{The complete bipartite graph $K_{n,m}$}

For $n=1$, the complete bipartite graph $K_{1,m}$ has $B(K_{1,m})=\frac{m+1}{2}$.
Each $\ell_1$-Fiedler vector for $K_{1,m}$ corresponds to an ordered quasi-bipartition $(P,N)$, where one of $P$ and $N$ is the vertex of high degree, and the other one is a non-empty subset of vertices from the $m$ vertices of the other bipartition class. 
The number of $\ell_1$-Fiedler vectors for $B(K_{1,m})$ is then
\[2^{m+1}-2.\]
For $\min\{n,m\}\geq 2$, $K_{n,m}$ has $B(K_{n,m})=\max\{n,m\}$. 
Assume $n>m>1$. Each $\ell_1$-Fiedler vector for $K_{n,m}$ corresponds to an ordered quasi-bipartition $(P,N)$, where each of  $P$ and $N$ is a non-empty subset of vertices from the bipartition class of $m$ vertices. For each of these $m$ vertices, there are three options, whether to assign it to $P$, $N$ or $0$. We then subtract the number of options where only $P$ and $0$ are used, or only $N$ and $0$ are used. We find that the number of $\ell_1$-Fiedler vectors for $B(K_{n,m})$ is
\[3^{m}-2^{m+1}+1.\]
Assume then that $n=m>1$. $K_{n,n}$ has $B(K_{n,n})=n$. 
Each $\ell_1$-Fiedler vector for $K_{n,n}$ corresponds to an ordered quasi-bipartition $(P,N)$, where each of $P$ and $N$ is a non-empty subset of vertices from the same bipartition class. 

If $P$ and $N$ belong to the same bipartiton class, then we have $2(3^n-2^{n+1}+1)$ $\ell_1$-Fiedler vectors, similar to the previous case. 

If $P$ and $N$ belong to different bipartition classes, we have  $2(2^n-1)^2$ $\ell_1$-Fiedler vectors. 

The number of $\ell_1$-Fiedler vectors for $B(K_{n,n})$  is then 
\[2(3^n-2^{n+1}+1) +  2(2^n-1)^2
= 2\left(4^n + 3^n -2^{n+2}+2  \right).\]

\section{A hardness result}\label{sec:hardness}

The results of the previous section show that, for several natural graph families, the set of
$\ell_1$-Fiedler vectors in $\mathcal G$ can be explicitly described and counted in closed form.
This naturally leads to the following general question: given a graph, how difficult is it to determine
the number of $\ell_1$-Fiedler vectors?

It is worth highlighting that, on the minimization side, Andrade and Dahl \cite{andrade2024combinatorial} already showed that the computation of $b(G)$ and a corresponding $\ell_1$-Fiedler vector is NP-hard, via the connection
between $b(G)$ and the sparsest cuts. 

In this section, we focus on the maximization side. Although Theorem \ref{theo:two_highest_geg} shows that the value of $B(G)$ admits a very simple expression, the structure of the subset of vectors of $\mathcal{G}$ that attain that value is not trivial in general; the associated counting problem is computationally intractable. 
\begin{theorem}
    The problem of counting the number of $\ell_1$-Fiedler vectors in $\mathcal G$ for the parameter $B(G)$ is \#P-complete.
\end{theorem}

\begin{proof}
The problem is in \#P, since given a graph $G=(V,E)$ and a candidate vector $\bm{g}\in\mathcal{G}$, one can verify in polynomial time whether it attains the maximum $f(\bm{g})=B(G)$ using Proposition~\ref{prop:sum_xis} and Theorem~\ref{theo:two_highest_geg}.

To prove \#P-hardness, we give a polynomial-time counting reduction from the problem of counting non-empty independent sets in $3$-regular graphs, which is \#P-complete by Greenhill~\cite{Greenhill2000ComplexityCounting}. Let $G$ be a $3$-regular graph on $n>4$ vertices. Construct $G'$ by adding a new vertex $z$ adjacent to every vertex of $G$. Then $\deg_{G'}(z)=n$, while every vertex of $G$ has degree $4$ in $G'$. Hence $B(G')=\frac12(n+4)$.

We claim that the optimal ordered quasi-bipartitions of $G'$ are precisely the pairs $(\{z\},I)$ and $(I,\{z\})$, where $I$ is a non-empty independent set of $G$.

On the one hand, if $(P,N)$ is optimal for $G'$, we must have
\[\xi(P)+\xi(N)=\operatorname{deg}_{\max}(P)+\operatorname{deg}_{\max}(N)= n+4.\]
Since one of the two sets must contain the unique degree-$n$ vertex $z$, the first equality forces that set to be exactly $\{z\}$. The other set must consist of degree-$4$ vertices and must be pairwise non-adjacent; equivalently, it is a non-empty independent set of the original graph $G$.

Conversely, if $I$ is a non-empty independent set of $G$, then in $G'$ we have $\xi(\{z\})=n$ and $\xi(I)=4$. Therefore
\[\frac12\bigl(\xi(\{z\})+\xi(I)\bigr)=\frac12(n+4)=B(G'),\]
so both $(\{z\},I)$ and $(I,\{z\})$ are optimal.

Thus the number of optimal points of $\mathcal G$ for $B(G')$ is twice the number of non-empty independent sets of $G$. Dividing by $2$ gives the desired polynomial-time counting reduction, and the result follows.
\end{proof}

\bibliographystyle{abbrv}

\bibliography{refs}

@article{andrade2024combinatorial,
  title={Combinatorial {F}iedler theory and graph partition},
  author={Andrade, Enide and Dahl, Geir},
  journal={Linear Algebra and its Applications},
  volume={687},
  pages={229--251},
  year={2024},
  publisher={Elsevier}
}

@article{fiedler1973algebraic,
  title={Algebraic connectivity of graphs},
  author={Fiedler, Miroslav},
  journal={Czechoslovak mathematical journal},
  volume={23},
  number={2},
  pages={298--305},
  year={1973},
  publisher={Institute of Mathematics, Academy of Sciences of the Czech Republic}
}

@book{chung1997spectral,
  author    = {Chung, Fan R. K.},
  title     = {Spectral Graph Theory},
  series    = {CBMS Regional Conference Series in Mathematics},
  volume    = {92},
  publisher = {American Mathematical Society},
  address   = {Providence, RI},
  year      = {1997}
}

@article{chang2016one,
  author  = {Chang, K. C. and Shao, Sihong and Zhang, Dong},
  title   = {The \(1\)-{Laplacian} {Cheeger} cut: Theory and algorithms},
  journal = {Journal of Computational Mathematics},
  volume  = {33},
  number  = {5},
  pages   = {443--467},
  year    = {2015}
}

@article{doehlert1972experimental,
  author  = {Doehlert, D. H. and Klee, V. L.},
  title   = {Experimental designs through level reduction of the \(d\)-dimensional cuboctahedron},
  journal = {Discrete Mathematics},
  volume  = {2},
  number  = {4},
  pages   = {309--334},
  year    = {1972},
  doi     = {10.1016/0012-365X(72)90011-8}
}

@article{ardila2010root,
  author  = {Ardila, Federico and Beck, Matthias and Ho{\c{s}}ten, Serkan and Pfeifle, Julian and Seashore, Kim},
  title   = {Root polytopes and growth series of root lattices},
  journal = {SIAM Journal on Discrete Mathematics},
  volume  = {25},
  number  = {1},
  pages   = {360--378},
  year    = {2011},
  doi     = {10.1137/090749293}
}

@misc{KannanRoy2026,
author = {Rajesh~Kannan, M. and Roy, R.},
title = {Structural and extremal properties of $\ell_1$-{F}iedler value. {A}rxiv.org/pdf/2601.05771},
year = {2026},
publisher = {Arxiv.org/pdf/2601.05771},
}

@article{anderson1985eigenvalues,
  author  = {Anderson, William N. and Morley, Thomas D.},
  title   = {Eigenvalues of the {Laplacian} of a graph},
  journal = {Linear and Multilinear Algebra},
  volume  = {18},
  number  = {2},
  pages   = {141--145},
  year    = {1985},
  doi     = {10.1080/03081088508817681}
}

@article{zhang2011laplacian,
  author  = {Zhang, Xiao-Dong},
  title   = {The {Laplacian} eigenvalues of graphs: {A} survey},
  journal = {arXiv preprint arXiv:1111.2897},
  year    = {2011}
}

@article{hall1970r,
  title={An $r$-dimensional quadratic placement algorithm},
  author={Hall, Kenneth M},
  journal={Management science},
  volume={17},
  number={3},
  pages={219--229},
  year={1970},
  publisher={INFORMS}
}

@article{koren2005drawing,
  title={Drawing graphs by eigenvectors: {T}heory and practice},
  author={Koren, Yehuda},
  journal={Computers \& Mathematics with Applications},
  volume={49},
  number={11-12},
  pages={1867--1888},
  year={2005},
  publisher={Elsevier}
}

@book{brouwer2011spectra,
  title={Spectra of graphs},
  author={Brouwer, Andries E and Haemers, Willem H},
  year={2011},
  publisher={Springer Science \& Business Media}
}

@article{Greenhill2000ComplexityCounting,
  author  = {Greenhill, Catherine},
  title   = {The complexity of counting colourings and independent sets in sparse graphs and hypergraphs},
  journal = {Computational Complexity},
  volume  = {9},
  number  = {1},
  pages   = {52--72},
  year    = {2000},
  month   = nov
}

@book{graham1994concrete,
  author    = {Ronald L. Graham and Donald E. Knuth and Oren Patashnik},
  title     = {Concrete Mathematics: A Foundation for Computer Science},
  edition   = {2},
  year      = {1994},
  publisher = {Addison-Wesley},
  isbn      = {0-201-55802-5}
}

@article{Prodinger82,
  title={Fibonacci numbers of graphs},
  author={Prodinger, H. and Tichy, R. F.},
  journal={The Fibonacci Quaterly},
  volume={20},
  number={1},
  pages={16--21},
  year={1982}
}

@article{Fukuda2004,
  author  = {Komei Fukuda},
  title   = {From the Zonotope Construction to the {M}inkowski Addition of Convex Polytopes},
  journal = {Journal of Symbolic Computation},
  volume  = {38},
  number  = {4},
  pages   = {1261--1272},
  year    = {2004},
  doi     = {10.1016/j.jsc.2003.08.007}
}

\end{document}